\documentclass{article}
\usepackage{cite}
\usepackage{amsmath,amssymb,amsfonts,amsthm}
\usepackage{algorithmic}
\usepackage{graphicx}
\usepackage{textcomp}
\usepackage{xcolor}
\usepackage{bbm}
\usepackage{optidef}
\usepackage{enumerate}
\usepackage{subcaption}

\usepackage{authblk}

\begin{document}

\newcommand{\ones}{\mathbf{1}}
\newcommand{\infnorm}[1]{\lVert{#1}\rVert_{\infty}}
\newcommand{\inv}{^{-1}}
\newcommand{\binv}{M}
\newcommand{\sat}[1]{\text{sat}\left({#1}\right)}
\newcommand{\dz}[1]{\text{dz}\left({#1}\right)}
\newcommand{\cu}{u}

\newcommand*{\revision}{\textcolor{black}}

\newtheorem{problem}{Problem}
\newtheorem{theorem}{Theorem}
\newtheorem{lemma}{Lemma}
\newtheorem{remark}{Remark}
\newtheorem{definition}{Definition}
\newtheorem{assumption}{Assumption}
\newtheorem{conjecture}{Conjecture}

\title{Anti-windup Coordination Strategy Around a Fair Equilibrium in Resource Sharing Networks}
\date{}

\author[1]{Felix Agner\thanks{Corresponding author: felix.agner@control.lth.se}}
\author[2]{Pauline Kergus}
\author[1]{Anders Rantzer}
\author[3]{Sophie Tarbouriech}
\author[3,4]{Luca Zaccarian}

\affil[1]{Department of Automatic Control, Lund University, Sweden}
\affil[2]{LAPLACE, Université de Toulouse, CNRS, INPT, UPS, Toulouse, France}
\affil[3]{LAAS-CNRS, Université de Toulouse, CNRS, Toulouse, France}
\affil[4]{Department of Industrial Engineering, University of Trento, Trento, Italy}


\maketitle

\begin{abstract}
We coordinate interconnected agents where the control input of each agent is limited by the control input of others. In that sense, the systems have to share a limited resource over a network. Such problems can arise in different areas and it is here motivated by a district heating example. When the shared resource is insufficient for the combined need of all systems, the resource will have to be shared in an optimal fashion. In this scenario, we want the systems to automatically converge to an optimal equilibrium. The contribution of this paper is the proposal of a control architecture where each separate system is controlled by a local PI controller. The controllers are then coordinated through a global rank-one anti-windup signal. It is shown that the equilibrium of the proposed closed-loop system minimizes the infinity-norm of stationary state deviations. \revision{A proof of linear-domain passivity is given, and a numerical example highlights the benefits of the proposed method with respect to the state-of-the-art.}

\end{abstract}


\section{Introduction}
In this paper we consider the problem of \revision{asymptotically} coordinating a large number of agents that share a central, limited resource towards an optimal equilibrium. Such problems arise in many applications, e.g. optimal power flow\cite{opf_pursuit_DallAnese}, the TCP protocol \cite{internet_congestion_lowetal, fairness_internet_congestion_kelly} and buffer networks \cite{BAUSO_asymptotical_optimality, BLANCHIN_network_decentralized_optimiziation, fair_network_flow_control}. \revision{We consider the motivating example of a district heating network where unfair situations can arise in peak load conditions; buildings close to heat sources stay warm but peripheral buildings become cold. Coordinating central buildings to reduce their heat load in these scenarios would yield a more fair heat distribution \cite{AGNER2022100067}. For a more detailed view on district heating systems and challenges in district heating control, see e.g. \cite{thebible, vandermeulen_controlling_2018}.}
\revision{
We consider a representation of such systems given by a linear system with saturating control:
\begin{equation}
    \Dot{x} = -x + B\sat{u} + w.
    \label{eq:system description}
\end{equation}
Here $x \in (x_1,\dots,x_n) \in \mathbb{R}^n$ represents deviations from reference levels for the agents, $w\in\mathbb{R}^n$ is a constant disturbance acting on the system, $u\in\mathbb{R}^n$ represents the control actions of the agents and $B\in\mathbb{R}^{n\times n}$ represents the interconnection among the agents. The saturation function $\sat{\cdot}$ represents the limited nature of the resource in the system. A more detailed description of the system will be given later. Problems of this form are addressed in \cite{BAUSO_asymptotical_optimality}, which shows that feedback control on the form $u = -B^Tx$ asymptotically minimizes the cost $x^Tx+v^Tv$ for  \eqref{eq:system description} where $v = \sat{u}$. Furthermore \cite{ fair_network_flow_control} designs a controller that asymptotically minimizes varying norms of $u$ in the non-saturated formulation of \eqref{eq:system description}. We extend the asymptotically optimal control design of these previous authors in three ways. First, we consider minimizing the cost function $\infnorm{x} = \max_i{|x_i|}$, associated with worst-case fairness. In the district heating example, this objective captures the deviation in the coldest building, which for the specific application is more important than minimizing $u$. Secondly, we approach scalability of the control strategy in another fashion. Indeed, \cite{BAUSO_asymptotical_optimality, fair_network_flow_control} approach scalability by considering systems where $B$ has a sparse structure, so that with  $u = -B^Tx$, each agent acts on a few measurements. We consider another scalable control approach of rank-one coordination as utilized in \cite{daria, carolina}. In this scenario, the signals from all the agents are combined into one scalar value, and then redistributed to the agents. The advantage of rank-one schemes is that they allow for implementations that scale well and maintain privacy among the agents, as well as scalability, even when $B$ is not sparse. We do this for a specific set of systems \eqref{eq:system description} where $B$ is an \textit{M-matrix}, a property that was not previously exploited in this context. Finally, we propose a control law where the results rely only on the structure of $B$, thus making the implementation robust to modeling errors.
}
\revision{
To solve the problem under consideration we propose a controller where each agent maintains a local proportional-integral-controller, and coordination is performed through a global rank-one anti-windup correction. Anti-windup techniques have a long-standing tradition of effective use in combination with integral controllers to improve performance\cite{anti-windup-tutorial}. However, recent results show a strong connection between anti-windup schemes and optimization \cite{aw_converges_to_projection, aw_implementation_for_optimization}, opening the possibility of considering anti-windup loops for optimal equilibrium coordination, in line with what we propose here.}

In this article we provide several contributions: We propose the aforementioned controller for driving the system to an optimal equilibrium. We show that under certain conditions on the disturbance $w$, such an equilibrium exists, is unique, and is in fact \revision{uniquely} optimal. Analytical proofs of convergence are left outside the scope of this work \revision{as we have not yet been able to demonstrate them. However, we provide sufficient conditions for stability in the linear domain $\sat{u}=u$. This leads us to formulate a stability conjecture, subject of future work. A numerical experiment is included to show the effectiveness of the proposed method.}

The paper is organized as follows. The system and problem under consideration are formally introduced in Section II, along with the proposed control scheme. The existence of an equilibrium for this system is considered in Section III, and the optimality of the system equilibrium is treated in Section IV. Section V introduces a conjecture on the convergence properties of the proposed closed loop, based on a proof of linear-domain stability under suitable conditions on the PI gains. \revision{A numerical example is shown in Section VI}. Finally conclusions and future work are covered in Section VII.

\textit{Notation}: If $A$ is a matrix then denote $A_i$ to be row $i$ of $A$ and $A_{i,j}$ be the element of $A$ at row $i$ and column $j$. Let $\ones$ be a column vector of all ones with dimensions taken from context, and thus $\ones \ones^T$ is a matrix of all 1's. Denote $\sat{\cdot}$ to be the saturation function $\sat{u} = \max \left( \min \left( u, 1 \right), -1 \right)$ and denote the dead-zone function $\dz{u} = u - \sat{u}$. With a slight abuse of notation, the dead-zone and saturation functions applied to vectors operate element-wise. Let the superscript $x^0$ denote the state $x$ in an equilibrium point, and the superscript $x^*$ to be the value of $x$ which solves an optimization problem. Let the infinite norm $\infnorm{\cdot}$ of a vector $v$ be defined as the maximum magnitude element $\max_i |v_i|$.

\section{System Description and Problem Formulation}

\subsection{System Description}
Consider system \eqref{eq:system description} where $x_i \in \mathbb{R}$ is the state of each agent $i = 1, \dots, n$ and $u_i \in \mathbb{R}$ is the controller output of each agent. \revision{ $B \in \mathbb{R}^{n \times n}$ is an \textit{M-matrix} \cite{horn_johnson_1991, tutorial_on_positive_systems}. Such matrices have non-positive off-diagonal entries, thus capturing the fact that each agent negatively impacts the others. Thus, if agent $i$ increases its control input it receives more resources and the other agents receive less resources. Denote $\binv = B\inv$. M-matrices have non-negative inverses in general and we will assume that $\binv$ is strictly positive, thus $\binv_i \ones > 0$ for all $i$. Input $w$ denotes a constant, unknown disturbance affecting the system. In practice the disturbance does not need to be constant, but sufficiently slowly varying.} Let us also introduce the index $k$ as a maximizing argument of the following expression
\begin{equation}
    k \in \mathcal{K} = \text{arg} \max_i \left| \frac{\dz{\binv_iw}}{\binv_i\ones}\right|,
    \label{eq:maximizing index k}
\end{equation}
which is, in general, nonunique and characterizes the agent that is most affected by the disturbance $w$.

\subsection{Problem Formulation}
We address here the unfair allocation of resources when the agents try to reject a constant disturbance that is too large to drive the system to the origin in view of input saturation. \revision{We consider a notion of fairness as described in \cite{fair_equilibrium} where \textit{"no individual can improve its performance without affecting at least one user adversely"} with regards to the deviations $x_i$, which we formally describe below.
}
\begin{definition} \label{def:fair equilibrium}
    \revision{
    An equilibrium pair ($x^0$, $u^0$) is \textit{fair} if there is no other equilibrium pair ($x^\dagger$, $u^\dagger$) where $\infnorm{x^\dagger} < \infnorm{x^0}$, or $\infnorm{x^\dagger} = \infnorm{x^0}$ and $|x^\dagger_i| < |x^0_i|$ for some $i$.
    }
\end{definition}
We therefore consider the following problem formulation.
\begin{problem}
    Design a feedback controller driving system \eqref{eq:system description} from any suitable initial condition to an equilibrium pair ($x^0$, $u^0$), such that $x^* = x^0$ and $\revision{v}^* = \sat{u^0}$ solves the optimization problem
    \begin{mini!} 
    {x,\revision{v}}{\infnorm{x}}
    {\label{eq:optimal problem}}{\label{eq:optimal cost}}
    \addConstraint{-x+B\revision{v}+w = 0} \label{eq:optimization equality}
    \addConstraint{-1 \leq \revision{v}\leq 1} \label{eq:optimization constrained u}
    \end{mini!}
    uniquely.
\end{problem}
\revision{As $\infnorm{x_0}$ is minimized uniquely, this equilibrium must be fair by Definition \ref{def:fair equilibrium}.
}
\subsection{Proposed Feedback Controller}
We propose an individual PI controller for each agent. Each controller has an integral state $z_i$, and strictly positive gains $p_i$ (proportional gain) and $r_i$ (integral gain). These gains can be tuned locally by the agents. We then introduce a scalar communication signal exchanged among the agents, resulting in a rank-one anti-windup correction term such that each controller adds the sum of all agents' dead-zones to their integrator input. The full closed-loop system can be written as
\begin{subequations}
\begin{align}
    \Dot{x} &= -x + B\sat{u} + w \label{eq:closed loop x}\\
    \Dot{z} &= x + \beta\ones\ones^T\dz{u} \label{eq:closed loop z}\\
    u &= -Px - Rz, \label{eq:closed loop u}
\end{align}
\label{eq:closed loop system}
\end{subequations}
where $P$ and $R$ are diagonal, positive matrices gathering the controller gains $p_i$, $r_i$, 
$\ones\ones^T\dz{u}$ is the rank-one anti-windup signal and $\beta$ is a positive, scalar anti-windup gain. One advantage of the proposed structure is that, under normal circumstances, each PI controller is completely disconnected from the other ones and acts based on local information only. If saturation occurs, the central signal activated and a fairness-oriented coupling emerges from the anti-windup term. Another advantage of the architecture \eqref{eq:closed loop system} is that when coupling occurs, the coupling signal is merely the sum of the dead-zones for each agent which can be computed efficiently. The summation hides the individual signals, so that when this central signal is redistributed to the agents, each agent does not know the dead-zone values for any of the other individual agents. As such, this global signal lends itself well to scalable and privacy-compliant implementations.

\section{Closed-Loop Equilibria}
In this section we characterize the equilibria of the proposed closed-loop system \eqref{eq:closed loop system}. From \eqref{eq:closed loop system}, any equilibrium $(x^0, z^0)$ solves the equations
\begin{subequations}
\begin{align}
    0 &= -x^0 + B\sat{u^0} + w \label{eq:equilibrium x}\\
    0 &= x^0 + \beta \ones\ones^T\dz{u^0} \label{eq:equilibrium z}\\
    u^0 &= -Px^0 - Rz^0. \label{eq:equilibrium u}
\end{align}
\label{eq:equilibrium}
\end{subequations}
It is not trivial to show whether a solution to \eqref{eq:equilibrium} exists. This section studies conditions for the existence and uniqueness of such solutions. Note that it is sufficient to study pairs ($x^0$, $\cu^0$) satisfying \eqref{eq:equilibrium x} and \eqref{eq:equilibrium z}, because the positive definiteness of $R$ implies its invertibility. Hence for any such state-control pair $(x^0,\cu^0)$ satisfying \eqref{eq:equilibrium x} and \eqref{eq:equilibrium z}, $z^0$ can be uniquely determined from \eqref{eq:equilibrium u}.

\subsection{Existence of an Equilibrium Point}
\revision{Recall that $\binv = B\inv$. We provide below a necessary and sufficient condition for \eqref{eq:closed loop system} to admit an equilibrium.}
\begin{lemma}
The closed-loop system \eqref{eq:equilibrium} admits an equilibrium point ($x^0$, $z^0$), if and only if
\begin{equation}
    \max_i \frac{\binv_iw-1}{\binv_i\ones} \leq \min_j \frac{\binv_jw+1}{\binv_j\ones}.
    \label{eq:condition on w}
\end{equation}
\end{lemma}

\begin{proof}
Let us begin with showing that \eqref{eq:condition on w} is necessary for the existence of an equilibrium point. \eqref{eq:equilibrium x} and \eqref{eq:equilibrium z} can be combined to
\begin{equation}
    \binv w + \sat{\cu^0} = -\beta \binv \ones \ones^T \dz{\cu^0}.
    \label{eq:equilibrium sat and dz equality}
\end{equation}
Thus
\begin{equation}
    \frac{\binv_i w + \sat{\cu^0_i}}{\binv_i \ones} = -\beta \ones^T \dz{\cu^0}, \quad \forall i= 1,\dots,n.
    \label{eq: equilibrium equality for each i}
\end{equation}
If \eqref{eq:condition on w} does not hold, then there exist $i$ and $j$ such that
\begin{equation}
    \frac{\binv_i w - 1}{\binv_i \ones} > \frac{\binv_j w + 1}{\binv_j \ones}.
    \label{eq:i > j}
\end{equation}
However, \eqref{eq: equilibrium equality for each i} implies that
\begin{equation}
     \frac{\binv_i w + \sat{\cu^0_i}}{\binv_i \ones} = \frac{\binv_j w + \sat{\cu^0_j}}{\binv_j \ones}.   
     \label{eq:i = j}
\end{equation}
As $\sat{\cu^0_i} \geq -1$ and $\sat{\cu^0_j} \leq 1$, \eqref{eq:i > j} and \eqref{eq:i = j} cannot simultaneously hold, which establishes a contradiction thus proving that there is no equilibrium. This proves that \eqref{eq:condition on w} is necessary for the existence of an equilibrium. For the sufficiency, first recall the definition of $k$, given by \eqref{eq:maximizing index k}. Then consider the candidate equilibrium $x^0$, $\cu^0$ given by
\begin{subequations}
\begin{align}
    x^0 &= \ones \frac{\dz{\binv_k w}}{\binv_k \ones} & \label{eq:candidate x} \\
    \cu^0_k &= -\sat{\binv_k w} - \frac{\dz{\binv_k w}}{\beta \binv_k \ones} &\label{eq:candidate uk} \\
    \cu^0_i &= - \binv_i w + \frac{\binv_i \ones}{\binv_k \ones}\dz{\binv_k w} , &\forall i \neq k.\label{eq:candidate ui} 
\end{align}
\label{eq:equilibrium candidate}
\end{subequations}
We show below that when \eqref{eq:condition on w} holds, the candidate equilibrium \eqref{eq:equilibrium candidate} solves \eqref{eq:equilibrium}. Consider 3 scenarios. (i): $\dz{\binv_k w} = 0$, (ii): $\dz{\binv_k w} > 0$ and (iii): $\dz{\binv_k w} < 0$. In scenario (i), $x^0 = 0$, and $\cu^0 = - \binv w$. As $\dz{\binv_k w} = 0$ in this scenario, $\dz{\binv_i w} = 0$ for all $i$. Otherwise \eqref{eq:maximizing index k} would not be maximized by $k$. This implies that $\sat{\cu^0} = \cu^0 = - \binv w$ and $\dz{\cu^0} = 0$. It is thus easy to verify that \eqref{eq:equilibrium} holds. In scenario (ii), note that the left side of \eqref{eq:condition on w} is maximized by index $k$ and can be reformulated as
\begin{equation}
    \frac{\dz{\binv_kw}}{\binv_k\ones} \leq \frac{\binv_i w+1}{\binv_i\ones} \quad \forall i=1,\dots,n. \label{eq:bound on ui}
\end{equation}
Returning to the candidate equilibrium and \eqref{eq:candidate ui} for $i\neq k$, 
\begin{equation}
        \cu^0_i = -\binv_i w + \frac{\binv_i \ones}{\binv_k \ones}\dz{\binv_k w} \leq 1
\end{equation}
where the inequality is derived from \eqref{eq:bound on ui}. Thus $\cu^0_i \leq 1$. In addition,
\begin{equation}
\begin{split}
        \cu^0_i &= -\binv_i w + \frac{\binv_i \ones}{\binv_k \ones}\dz{\binv_k w} \\
       &=-\sat{\binv_i w } - \dz{\binv_i w }+ \frac{\binv_i \ones}{\binv_k \ones}\dz{\binv_k w} \\
       &=-\sat{\binv_i w } + \binv_i \ones \left(\frac{\dz{\binv_k w}}{\binv_k \ones} - \frac{\dz{\binv_i w}}{\binv_i \ones} \right) \\
       & \geq -1,
\end{split}
\end{equation}
where the last inequality holds because $k$ maximizes \eqref{eq:maximizing index k}. This means that $-1 \leq \cu^0_i \leq 1$ for all $i \neq k$. Thus 
\begin{equation}
    \sat{\cu^0_i} = \cu^0_i = - \binv_i w + \frac{\binv_i \ones}{\binv_k \ones}\dz{\binv_k w}, \quad \forall i \neq k,\label{eq:sat ui0}
\end{equation}
and
\begin{equation}
    \dz{\cu_i} = 0, \quad \forall i \neq k.
    \label{eq:dz ui0}
\end{equation}
For index $k$, \eqref{eq:candidate uk} provides
\begin{equation}
    \sat{\cu^0_k} = -\sat{\binv_k w} = - \binv_k w + \dz{\binv_k w}
    \label{eq:sat uk0}
\end{equation}
and
\begin{equation}
    \dz{\cu^0_k} = - \frac{\dz{\binv_k w}}{\beta \binv_k \ones}.
    \label{eq:dz uk0}
\end{equation}
Combining \eqref{eq:sat ui0}, \eqref{eq:dz ui0}, \eqref{eq:sat uk0} and \eqref{eq:dz uk0} yields
\begin{equation}
    \sat{\cu^0} = -\binv w + \binv \ones \frac{\dz{\binv_k w}}{\binv_k \ones}
    \label{eq:scenario ii sat}
\end{equation}
and 
\begin{equation}
    \ones^T \dz{\cu^0} = - \frac{\dz{\binv_k w}}{\beta \binv_k \ones}.
    \label{eq:scenario ii dz}
\end{equation}
which allows us to easily verify that ($x^0$, $u^0$) from \eqref{eq:equilibrium candidate} solves \eqref{eq:equilibrium} in scenario (ii). An equal argument can be made for scenario (iii), which we omit for brevity. This shows that given any scenario for $\dz{\binv_k w}$, the candidate equilibrium \eqref{eq:equilibrium candidate} is valid when \eqref{eq:condition on w} holds. Thus \eqref{eq:condition on w} is both necessary and sufficient for the existence of an equilibrium.
\end{proof}

To interpret \eqref{eq:condition on w}, note that it is satisfied when all entries $w_i$ are similar to each other. For instance, $w = s\ones$ for any scalar $s$ trivially satisfies the condition. This makes it a sensible assumption when the disturbance $w$ affects all agents in a similar way. This is for instance the case in the district heating example, where the outdoor temperature is likely to be quite similar for all the buildings located in a specific area. To simplify our follow-up definitions, we will assume that \eqref{eq:condition on w} holds with a strict inequality, as formulated below.
\begin{assumption} \label{ass:w bounded}
The disturbance $w$ satisfies \eqref{eq:condition on w} strictly, namely
\begin{equation}
    \max_i \frac{\binv_iw-1}{\binv_i\ones} < \min_j \frac{\binv_jw+1}{\binv_j\ones}
    \label{eq:strict condition on w}.
\end{equation}
\end{assumption}
We assume the strict inequality to enforce uniqueness of the equilibrium, which is studied in the next section.

\subsection{Uniqueness of the Equilibrium}
Lemma 1 shows that under Assumption \ref{ass:w bounded}, there is an equilibrium for the closed-loop system. We study here conditions for this equilibrium to be unique. To enforce the uniqueness of this equilibrium, we assume the following.
\begin{assumption} \label{ass:k unique}
Either $\dz{\binv w} = 0$, or the maximizing argument $k$ given by \eqref{eq:maximizing index k} is unique.
\end{assumption}
\revision{If $k$ is non-unique, an arbitrarily small perturbation of $B$ or $w$ would make it so. In practical applications, $w$ is expected to vary slowly over time. This makes it unlikely that $k$ would be non-unique for an extended period of time, but may also cause $k$ to shift between agents. The analysis of such scenarios requires to study the transient behavior of the system, which is outside the scope of the paper, but will be the subject of future work.}

\begin{lemma}
    If Assumptions \ref{ass:w bounded} and \ref{ass:k unique} hold, then \eqref{eq:equilibrium candidate} is the unique equilibrium of the closed-loop system \eqref{eq:closed loop system}.
\end{lemma}
\begin{proof}
    Recall from the proof of Lemma 1 that for any equilibrium inducing input $\cu^0$, identity \eqref{eq:equilibrium sat and dz equality} must hold. Now denote
    \begin{equation}
        t = \beta  \ones^T \dz{\cu^0},
        \label{eq:definition of t}
    \end{equation}
    which allows \eqref{eq:equilibrium sat and dz equality} to be rewritten as
    \begin{equation}
        \sat{\cu^0_i} = -\binv_i w - M_i \ones t, \quad \forall i = 1,\dots,n.
        \label{eq:sat equality t}
    \end{equation}
    Note that if $t>0$, there must exist an $i \in \{ 1,\dots,n \}$ such that $\sat{\cu^0_i} = 1$. Similarly, if $t<0$, there exists an $i \in \{ 1,\dots,n \}$ such that $\sat{\cu^0_i} = -1$. Also note that \eqref{eq:strict condition on w} implies that there cannot exist $i$ and $j$ such that $\binv_i w \geq 1$ and $\binv_j w  \leq -1$. This in turn implies that either $\dz{\cu^0} \geq 0$ or $\dz{\cu^0} \leq 0$, where the inequality should be understood componentwise.
    Now, recalling that $k$ in \eqref{eq:maximizing index k} is unique by assumption, consider 3 scenarios; (i): $\dz{\binv_k w} = 0$, (ii): $\dz{\binv_k w} > 0$ and (iii): $\dz{\binv_k w} < 0$. In scenario (i), we see that $\dz{\binv_i w} = 0$ for all $i = 1,\dots,n$, as otherwise $|\dz{\binv_i w}| > 0$ for some $i$, implying that \eqref{eq:maximizing index k} would be maximized by this $i$. Thus $|M_i w| \leq 1$ for all $i$. Through \eqref{eq:sat equality t}, we prove next that this implies $t=0$. Indeed, assume by an absurd argument that $t>0$. Then \eqref{eq:sat equality t} yields $\sat{\cu^0_i} = -\binv_i w - M_i \ones t < -\binv_i w \leq 1$ for all $i=1,\dots,n$. But if $\sat{\cu^0_i} < 1$ for all $i$, then $\dz{\cu^0_i} \leq 0$ for all $i$ and thus $t$ cannot be positive. A parallel contradiction can be built for $t<0$. Thus we conclude that $t=0$. In turn, $t=0$ implies that $\dz{\cu^0}=0$, because \eqref{eq:definition of t} shows that $t$ is the sum of the entries of $\dz{\cu^0}=0$, multiplied by the positive scalar $\beta$. As the entries of $\dz{\cu^0}$ are either all positive or all negative, $t$ can only be 0 if all of the entries of $\dz{\cu^0}$ are 0. This uniquely fixes $\cu^0 = -\binv w$, which is the same as the candidate solution \eqref{eq:equilibrium candidate}. This in turn uniquely fixes $x^0$ through \eqref{eq:equilibrium z}, and uniquely fixes $z^0$ through \eqref{eq:equilibrium u}.

   In scenario (ii), \eqref{eq:sat equality t} implies $t \leq -\frac{\dz{\binv_k w}}{\binv_k \ones}$, because otherwise $\sat{\cu^0_k} < -1$. If $t = \frac{\dz{\binv_k w}}{\binv_k \ones}$ then $\sat{\cu^0_k} = -1$. For $i \neq k$,
    \begin{equation}
    \begin{split}
            \sat{\cu^0_i} &= -\binv_i w + \frac{\binv_i \ones}{\binv_k \ones}\dz{\binv_k w} \\
           &= -\sat{\binv_i w } - \dz{\binv_i w }+ \frac{\binv_i \ones}{\binv_k \ones}\dz{\binv_k w} \\
           &=-\sat{\binv_i w } + \binv_i \ones \left(\frac{\dz{\binv_k w}}{\binv_k \ones} - \frac{\dz{\binv_i w}}{\binv_i \ones} \right) \\
           &> -1.
    \end{split}
    \end{equation}
    The last inequality holds because $k$ uniquely maximizes \eqref{eq:maximizing index k} and $\binv_i \ones > 0$ due to non-negativity and invertibility of $\binv$. Thus we conclude that $t = -\frac{\dz{\binv_k w}}{\binv_k \ones}$, because otherwise $\sat{\cu^0_i} > -1$ for all $i$, contradicting the fact that $t$ is negative. For this scenario (ii), \eqref{eq:strict condition on w} can be written as
    \begin{equation}
        \frac{\dz{\binv_k w}}{\binv_k \ones} < \frac{\binv_i w + 1}{\binv_i \ones},\quad \forall i \neq k ,
    \end{equation}
    Which can be combined with $t = -\frac{\dz{\binv_k w}}{\binv_k \ones}$ to show that, for $i \neq k$,
    \begin{equation}
        \sat{\cu^0_i} = -\binv_i w + \frac{\binv_i \ones}{\binv_k \ones}\dz{\binv_k w}  < 1. \label{eq:mid-lemma-2-proof inequality}
    \end{equation}
    Inequality \eqref{eq:mid-lemma-2-proof inequality} implies $|\sat{\cu^0_i}| < 1$ for all $i\neq k$, and thus $\dz{\cu^0_i}=0$ for all $i\neq k$. This implies
    \begin{equation}
        t = \beta \ones^T \dz{\cu^0} = \beta \dz{\cu^0_k}
    \end{equation}
    and thus
    \begin{equation}
        \dz{\cu^0_k} = - \frac{\dz{\binv_k w}}{\beta \binv_k \ones}. \label{eq:lemma 2 dz k}
    \end{equation}
    Equations \eqref{eq:mid-lemma-2-proof inequality} and \eqref{eq:lemma 2 dz k} uniquely determine $\cu^0$, and, together with \eqref{eq:equilibrium candidate}, $x^0$ and $z^0$ are uniquely determined. For scenario (iii), a symmetric argument can be followed, which is omitted for brevity, thus completing the proof.
\end{proof}

\section{Optimality}
We proved in the previous section that under Assumptions \ref{ass:w bounded} and \ref{ass:k unique}, the proposed closed-loop system has a unique equilibrium, given by \eqref{eq:equilibrium candidate}. In this section, we will prove that this equilibrium is also the unique, optimal solution to \eqref{eq:optimal problem}.

\begin{theorem}
    If Assumptions \ref{ass:w bounded} and \ref{ass:k unique} hold, then $x^* = x^0$ and $\revision{v}^* = \sat{\cu^0}$ is \revision{the unique solution to} \eqref{eq:optimal problem} where $x^0$ and $\cu^0$ are given by \eqref{eq:equilibrium candidate}.
\end{theorem}

\begin{proof}
Lemma 2 proves that under Assumptions \ref{ass:w bounded} and \ref{ass:k unique}, $(x^0,\cu^0)$ is a state-input equilibrium pair. This means that $(x^*,\revision{v}^*)=(x^0, \sat{u^0})$ satisfies the constraints \eqref{eq:optimization equality} and \eqref{eq:optimization constrained u} and is therefore feasible. What remains is only to show that it is not only feasible but also \revision{uniquely} optimal. Consider for establishing a contradiction that there exists $\xi \neq 0$, such that $x^\dagger = x^*+\xi$, along with $\revision{v}^\dagger = \revision{v}^* + \binv \xi$ is also feasible and provides a lower or equal cost than $(x^*,\revision{v}^*)$. An equivalent rewriting of \eqref{eq:optimal problem} using $\xi$ is
\begin{mini!} 
{\xi}{\infnorm{\ones\frac{\dz{\binv_k w}}{\binv_k\ones} + \xi}}
{\label{eq:xi cost}}{}
\addConstraint{\infnorm{\binv (\xi - w + \ones \frac{\dz{\binv_k w}}{\binv_k\ones} )}  \leq 1.} \label{eq:xi restricted control}
\end{mini!}
First note that if $\dz{\binv_k w}=0$, then $\xi=0$ is trivially optimal as any $\xi\neq 0$ would yield a higher cost and thus not be an optimizer. Then consider the case where $\dz{\binv_k w} > 0$. For $\xi$ to provide a lower or equal cost, it must hold that $\xi_i \leq 0$ for all $i$. However, analyzing constraint \eqref{eq:xi restricted control} for index $k$ yields
\begin{equation}
    |-\binv_k w + \dz{\binv_k w} + \binv_k \xi | \leq 1.
    \label{eq:optimization constraint for k}
\end{equation}
Since we are focusing on the case $\dz{\binv_k w} > 0$, \eqref{eq:optimization constraint for k} reduces to 
\begin{equation}
    | -1 + \binv_k \xi | \leq 1.
    \label{eq:optimization constraint for k reduced}
\end{equation}
Since $\xi \leq 0$, inequality \eqref{eq:optimization constraint for k reduced} \revision{can only hold for $x_i=0$ as $\binv_k$ has strictly positive entries}. Therefore $\xi=0$ is uniquely optimal when $\dz{\binv_k w} > 0$. A parallel reasoning can be performed for the case $\dz{\binv_k w} < 0$. Thus $(x^*,\revision{v}^*)$ is the optimal solution to \eqref{eq:optimal problem}.
\end{proof}

\section{Stability Properties}
\revision{The results of Sections III and IV established that under Assumptions \ref{ass:w bounded} and \ref{ass:k unique}, the unique equilibrium of the closed-loop system \eqref{eq:closed loop system} solves the optimization problem \eqref{eq:optimal problem}. In this section we formulate the following conjecture regarding its stability properties.
}
\begin{conjecture}
\revision{
Under Assumptions \ref{ass:w bounded} and \ref{ass:k unique}, if $p_i > r_i$ for all $i=1,\dots,n$, then the proposed controller \eqref{eq:closed loop system} globally solves Problem 1.
}
\label{conj:convergence conjecture}
\end{conjecture}
\revision{
Conjecture \ref{conj:convergence conjecture}, subject to its proof, would provide strong properties for the proposed control law, granting stability and optimality for a large family of systems subject to a simple control tuning constraint. The proof however is non-trivial and requires results for saturated systems operating deeply in the saturated regime, which is why it is left outside the scope if this work. Our confidence in Conjecture \ref{conj:convergence conjecture} arises from numerous simulations of randomized systems. In addition, the specific choice of $p_i > r_i$ provides notions of stability for our problem through the following lemma. 
}
\begin{lemma} \label{lemma:linear stability}
\revision{
Assume that $p_i > r_i$ for all $i$ and $w\in \mathcal{L}_2$. Then system \eqref{eq:closed loop system} is asymptotically stable in the region of linearity where $\sat{u} = u$.
}
\end{lemma}
\begin{proof}
\revision{
Define $y=-Bu$. When $\sat{u} = u$ and thus $\dz{u}=0$, the closed loop system \eqref{eq:closed loop system} can be reformulated in the frequency domain as
\begin{subequations}
    \begin{align}
        sX &= -X - Y + W \\
        sU &= (P-R)X + PY - PW.
    \end{align}
\end{subequations}
These equations are fully diagonal, and can for each agent $i$ be combined to form
\begin{equation}
    U_i = \frac{r_i + p_i s}{s(s+1)}(Y_i - W_i) = G_i(s)(Y_i-W_i).
    \label{eq:linear transfer function}
\end{equation}
This feedback interconnection is represented in Figure \ref{fig:linear block diagram}. When $p_i > r_i$, the transfer function \eqref{eq:linear transfer function} is positive real, making it a passive component \cite{khalil}. In addition, due to $B$ being an M-matrix, we know that there exists a positive, diagonal matrix $D$ such that $-DB - B^T D \prec 0$. This means that the combined upper block of Figure \ref{fig:linear block diagram} is strictly passive. The multiplication by the positive, diagonal matrix $D\inv$ does not affect the passivity properties of $G_1(s)\dots G_n(s)$. The feedback interconnection between the strictly passive upper block and the passive lower block means that for any $w \in \mathcal{L}_2$, we have $u \in \mathcal{L}_2$\cite{feedback_systems}. This means that $\lim_{t \to \infty} w(t) = 0$, $\lim_{t \to \infty} u(t) = 0$ and thus clearly $\lim_{t \to \infty} x(t) = 0$ by \eqref{eq:closed loop x}.
}
\end{proof}

\begin{figure}
    \centering
    \vspace{8pt}
    \includegraphics[width = .6\columnwidth]{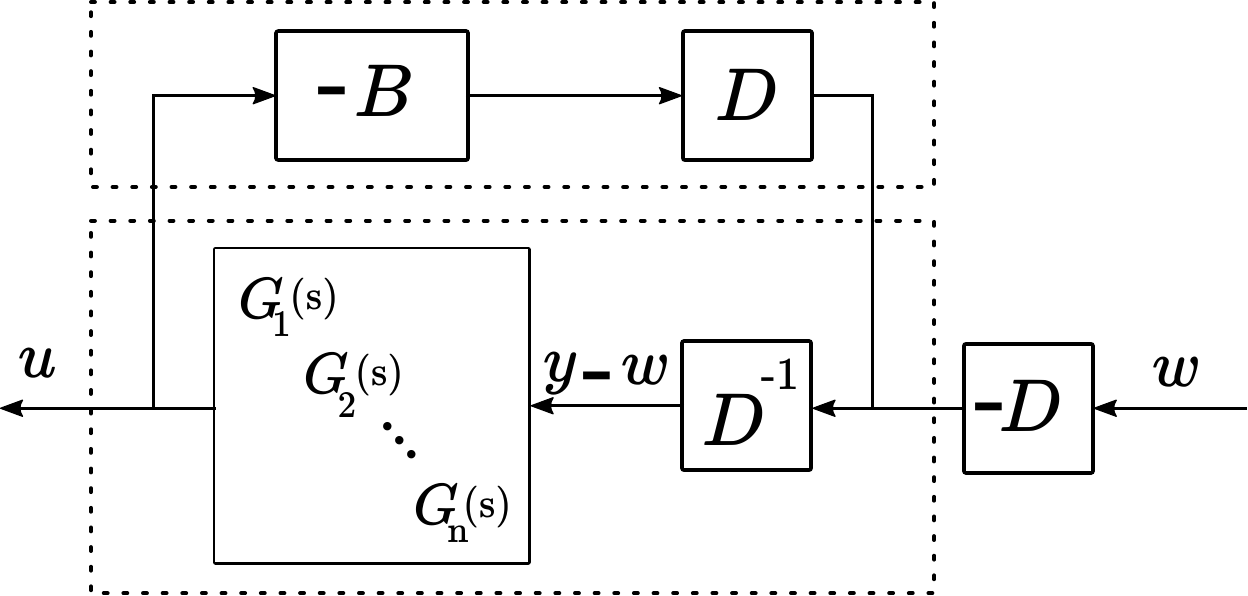}
    \caption{Block diagram, showing the interconnection used in the proof of Lemma \ref{lemma:linear stability}.}
    \label{fig:linear block diagram}
\end{figure}
\revision{
To prove or refute Conjecture \ref{conj:convergence conjecture} in future work, we believe that these passivity properties may be a useful tool. While it can be shown that the condition $P>R$ is conservative, we have also found examples of sufficiently large integral gains causing instability, thereby suggesting that our conjecture is reasonable.
}

\section{Numerical Example}

\begin{figure}[h] 
     \centering
     \includegraphics[width= 0.6\columnwidth]{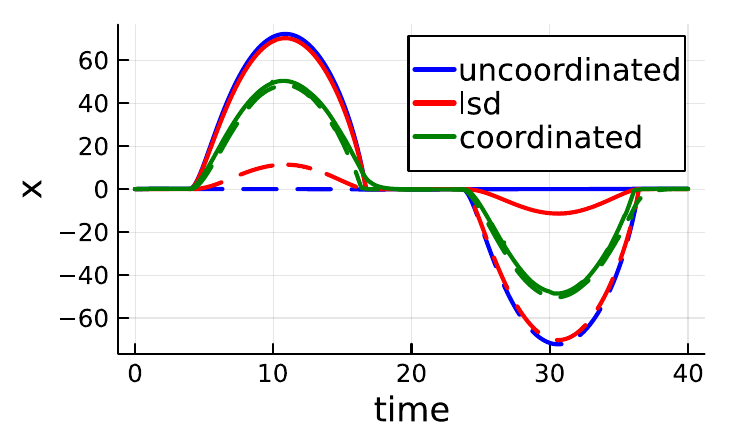}
     \caption{Envelopes of the states $x$ for each strategy. The dashed lines constitute the minimum $\min_i x_i(t)$ and the solid lines the maximum $\max_i x_i(t)$.}
     \label{fig:envelopes}
\end{figure}

\begin{figure}[h]
    \centering
    \includegraphics[width= 0.6\columnwidth]{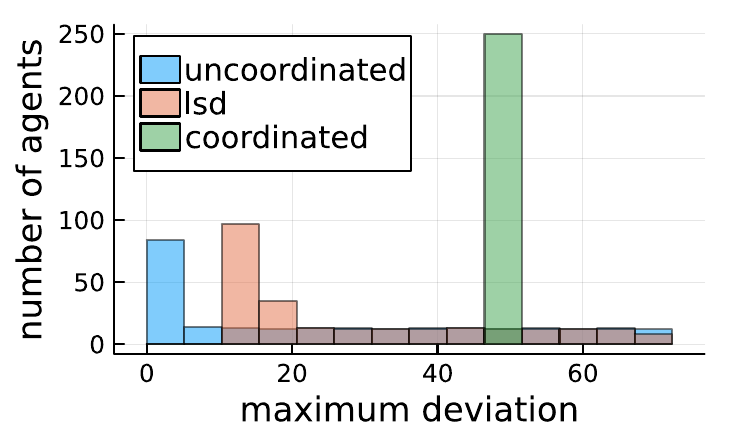}
    \caption{Histogram of maximum absolute deviations $\max_t |x_i(t)|$ experienced under each strategy. From deviation 20 to 80, the red and blue bars overlap.}
    \label{fig:histogram}
\end{figure}

\revision{To demonstrate the usefulness of the proposed controller, we investigate $n=250$ agents interconnected through the matrix $B=D(1.2n I - \ones \ones^T)$ where $D = \text{diag}(d_1,d_2, \dots , d_n)$ and $d_1, \dots , d_n$ are distributed at even intervals between 0.5 and 1.5. $w(t) = \ones \frac{n \sin{(t/2\pi)}}{2}$. We compare three strategies: First the \textit{coordinated} strategy, consisting in the controller proposed in this paper using the gains $p_i=1$, $r_i=1.5$ for all $i$ and $\beta=1$. Secondly the \textit{uncoordinated} strategy, namely the same PI-controllers as those of the coordinated case, only equipped with a local anti-windup action: $\dot{z}_i = x_i + \beta\dz{u_i}$. Finally the \textit{linear saturated decentralized} (lsd) controller $u = -B^Tx$ as proposed by \cite{BAUSO_asymptotical_optimality}. The systems are simulated using the \texttt{DifferentialEquations} toolbox \cite{rackauckas2017differentialequations} in Julia. Figure \ref{fig:envelopes} shows the envelopes of the time series over the simulation. In the coordinated case (green), all of the states $x$ are nearly completely synchronized. Under both the uncoordinated (blue) and the lsd (red) strategy, there is a large discrepancy between the maximum and minimum states. Furthermore, the lsd strategy is optimal with regards to a tradeoff between states $x$ and control action $u$, and therefore no states are driven to the origin with large disturbances $w$. Figure \ref{fig:histogram} shows histograms of the worst magnitude deviations in each strategy. We see that both the uncoordinated (blue) and lsd (red) strategies have several agents with larger deviations than any of the agents in the coordinated case. However, both the uncoordinated and lsd strategies also have many agents with lower deviations than that of the coordinated case.}

\section{Conclusion}
In this paper we have presented a controller for coordinating the control actions of agents that share a central resource. We proved that the only equilibrium of this closed-loop system is optimally fair. \revision{This optimality concerns the states $x$, an important extension of the literature which has mainly focused on properties of the control input $u$. A conjecture was proposed giving conditions for stability of this optimal equilibrium, motivated by passivity of the closed-loop system in the linear domain.}

Subject to the proof of Conjecture 1, the proposed method has many advantages. Each agent could tune the gains of a PI-controller locally while maintaining global guarantees of stability. These guarantees are only dependent on the structure of the system and not the model itself (i.e. the $B$-matrix does not have to be known, only that it has certain properties). \revision{The rank-one communication scheme ensures scalability of the implementation which does not require sparsity of $B$.}

Extensions of the work include exploiting the proven passivity property to prove stability with regards to the optimal equilibrium. Further system structures could be considered, for instance more general $A$-matrices, output feedback, or non-linear interconnections $B(u)$ which maintain similar properties to the current $B$-structure. Finally, one can consider analyzing and improving transient performance.

\section{Acknowledgment of Support}
The work of Felix Agner and Anders Rantzer is funded by the European Research Council (ERC) under the European Union's Horizon 2020 research and innovation program under grant agreement No 834142 (ScalableControl). They are members of the ELLIIT Strategic Research Area at Lund University. The work of Luca Zaccarian is supported in part by ANR via grant HANDY, number ANR-18-CE40-0010.


\bibliographystyle{ieeetr}
\bibliography{bibliography}

\end{document}